\documentclass[a4paper,11pt]{article}

\usepackage{geometry}

\usepackage{amsmath}
\usepackage{amssymb}
\usepackage{amsthm}
\usepackage {xcolor}

\usepackage{caption}
\usepackage{graphicx}

\newcommand{\rf}[1]{(\ref{#1})}

\newcommand{\beq}{\begin{equation}}
\newcommand{\eeq}{\end{equation}}
\newcommand{\bea}{\begin{eqnarray}}
\newcommand{\eea}{\end{eqnarray}}
\newcommand{\beas}{\begin{eqnarray*}}
\newcommand{\eeas}{\end{eqnarray*}}
\newcommand{\beqs}{\begin{displaymath}}
\newcommand{\eeqs}{\end{displaymath}}

\newcommand{\br}{\langle}
\newcommand{\kt}{\rangle}
\newcommand{\bra}[1]{\langle {#1}|}

\newcommand{\bdm}{\begin{displaymath}}
\newcommand{\edm}{\end{displaymath}}

\newcommand{\bbZ}{{\mathbb Z}}

\newcommand{\abs}[1]{\vert{ #1}\vert}
\newcommand{\half}{\frac{1}{2}}
\newtheorem{lem}{Lemma}

\newtheorem{thm}{Theorem}
\newtheorem{rem}{Remark}

\newcommand{\cut}{[-{\frac{a}{q}}, {\frac{a}{q}}]}

\newcommand{\Bbar}{ {\widetilde{G}} }

\title{\bf The Spectrum of  Asymptotic Cayley Trees}
\author{Bergfinnur Durhuus\footnote{\parbox[t]{\textwidth} 
{Department of Mathematical Sciences, University of Copenhagen, Universitetsparken 5, \newline DK 2100 Copenhagen, Denmark, e-mail: durhuus@math.ku.dk}}\\
Thordur Jonsson\footnote{The Science Institute, University of Iceland, 107 Reykjavik, Iceland, e-mail: thjons@hi.is} \\
John Wheater\footnote{\parbox[t]{\textwidth} {Rudolf Peierls Centre for Theoretical Physics, Department of Physics, Parks Road,\newline Oxford OX1 3PU, UK, e-mail: john.wheater@physics.ox.ac.uk}}}

\date{\today}

\begin{document}

\maketitle

\renewcommand{\baselinestretch}{1.1}

\abstract{We characterize the spectrum of the transition matrix for simple random walk 
on graphs consisting of a finite graph 
with a finite number of infinite Cayley trees attached.   We show that 
there is a continuous spectrum identical to that for a Cayley tree and, in general, a 
non-empty pure point spectrum. We apply
our results to studying continuous time quantum walk
on these graphs.  If the pure point spectrum is nonempty the walk is in general 
confined with a nonzero probability.}

\newpage
\section{Introduction}

The spectrum of the Laplacian and similar operators defined on graphs, both finite and infinite, has long been a topic of interest and many results are known, see for example \cite{graphsbook}. These systems are not only of intrinsic interest but also, for example, describe  the propagation of signals in discrete models of media or in networks of various kinds.  In this paper we study the spectrum of the transition, or hopping,  matrix on a class of infinite graphs that we call asymptotic Cayley trees.  These graphs, which are defined in  Section \ref{subsec:ACayleytree defn},  are trees
with a fixed coordination number outside a finite subgraph; they can be constructed by grafting planted Cayley trees by the root to the vertices of a finite graph.

The transition matrix is the adjacency matrix scaled by the degree of vertices and is therefore the generator of simple random walk on the graph. Furthermore, walks on an asymptotic Cayley tree can easily be decomposed into walks on the finite graph and walks on the grafted tree(s).  We exploit these two facts by using elementary probabilistic methods to find an exact rational relationship  between the resolvent (or Green's) functions of the graph of interest and those of the simpler constituents. These relationships enable us to characterize the spectrum of the transition matrix  on an asymptotic Cayley tree. The spectrum consists of a continuous component with support on the same interval as for the regular tree, and a pure point component. For the latter we 
establish bounds on the multiplicity of any given eigenvalue, and both upper and lower bounds on the total number of such normalizable eigenfunctions. These bounds are directly expressed in terms of the basic properties of the finite graph and its spectrum. 
The total number of normalizable eigenfunctions is shown to be bounded above by the number of vertices in the finite graph. The lower bound demonstrates that large classes of graph must have normalizable eigenfunctions. In particular 
if the number of eigenfunctions of the finite sub-graph having eigenvalues lying outside the continuous spectrum  exceeds twice the number of graft vertices, then the graph must have a non-empty pure point spectrum.
In addition to  these general theorems that apply to all asymptotic Cayley trees, we discuss as examples special cases where more precise statements can be made.

Our results are complementary to, and extend, those of \cite{verdiere} where the spectrum of the adjacency matrix on these graphs was studied by 
mapping the problem onto one of the adjacency matrix 
on a regular Cayley tree plus a finite rank perturbation.  
The mapping is shown always to exist but the relationship between the perturbation and the properties of the original graph is rather indirect.
The perturbation problem is then analysed using the apparatus of the theory of Schr\"odinger operators; there is a continuous component of the spectrum with support on the same interval as the regular tree and a pure point spectrum. The 
dimension of the pure point component is bounded above by the rank of the perturbation, and examples are given to show that normalizable eigenfunctions are indeed present in some cases. However, neither a lower bound on the dimension, nor results for the multiplicity of individual eigenvalues are given.

As an example of the application of our results we examine the properties of continuous time quantum walk on 
asymptotic Cayley trees. Quantum walks on graphs have been studied extensively in recent years, the 
main motivation being the development of efficient quantum algorithms.  For recent reviews
see, e.g., \cite{review,kadian2021quantum,Portugal}.  There are two classes of quantum walk: the discrete time walk which introduces a
new quantum degree of freedom usually called the `coin'; and the continuous time quantum walk which is
 governed by the Schr\"odinger equation and does not require extra  
 degrees 
of freedom.  The relation between these two types of quantum walks is complicated \cite{childs,strauch}, and seems not to be completely understood.  In this 
paper we are solely concerned with continuous time walks whose evolution  Hamiltonian is taken to be (minus) the transition matrix.

This paper is organized as follows.  In Section \ref{sec:graphsandwalks} we introduce our notation for graphs and the operators defined on them, 
and study classical random walk on the pure Cayley tree.   In Section \ref{sec:ACayleytree}
we define asymptotic Cayley trees, establish a number of results for  
classical random
walk on such graphs, and then prove our main theorems about their spectra.  In Section \ref{sec:examples} we discuss some 
special cases, and in Section \ref{sec:quantumwalks} we use the results from the earlier sections to analyse the behaviour
of continuous time quantum walks on  asymptotic Cayley trees.

\section{Graphs and random walks}\label{sec:graphsandwalks}

\subsection{Basic definitions}

Let $G=(V,E)$ be an undirected simple (i.e., with no multiple edges and no loops) 
connected graph with vertex set $V=V(G)$ and edge set $E=E(G)$. 
{For $u,v\in V$ we denote by $d_G(u,v)$ the graph distance between $u$ and $v$.}
We will consider both finite and infinite graphs but they are all locally finite, i.e., the number 
of neighbours $\sigma_u$ of any vertex $u$, called the degree of $u$, is assumed to be finite. 
In fact, for 
the graphs to be considered the degree is uniformly bounded,
\begin{equation}\label{eqn:degreebound}
\sigma_u \,\leq \, C\,,\quad  u\in V(G)\,,
\end{equation}
for some finite constant $C$. We consider the vector space $\ell^2(G)$ consisting of square summable 
complex valued functions $f:V\mapsto \mathbb C$ equipped with two different inner 
products, $\br\cdot|\cdot\kt$ and $(\cdot,\cdot)$, given by 
\beq\label{1}
\br f | g\kt =\sum_{u\in V} \bar{f}(u)g(u)\quad \mbox{and}\quad (f,g) = \sum_{u\in V} \sigma_u^{-1}\bar{f}(u)g(u)\,.
\eeq
Due to \eqref{eqn:degreebound}, the corresponding norms $\|\cdot\|_s$ and $\|\cdot\|_r$ are equivalent,
\beq
\|f\|_r\leq \|f\|_s\leq C\|f\|_r\,,
\eeq
and $\ell^2(G)$ is a Hilbert space with respect to both inner products. 
The functions in $\ell^2(G)$ will be referred to as $\ell^2$-functions.
For $u\in V$ we denote by $|u\kt$ the function which takes the value 1 at $u$ and is zero elsewhere, 
{and by
$\bra u$ its covector fulfilling $\br u|v\kt=\delta_{u,v}$ for all $v\in V$}.

Given a graph $G=(V,E)$, its adjacency matrix $A=(A_{uv})$, indexed by $V$, is defined by 
$A_{uv}=1$ if $(u,v)\in E$ and $A_{uv}=0$ otherwise.  The matrix 
$A$ defines an operator on $\ell^2(G)$ by
\beq\label{5}
(Af)(u)=\sum_{v:(u,v)\in E} f(v)\,.
\eeq
Using \eqref{eqn:degreebound}, it is easily seen that $A$ is a bounded operator and 
it is self-adjoint with respect to $\br\cdot|\cdot\kt$. 
The degree matrix is defined as the diagonal matrix $D_{uv}=\delta_{uv}\sigma_u$
and it likewise defines an operator $D$ on $\ell^2(G)$ which is bounded by 
\eqref{eqn:degreebound} and clearly has a bounded inverse. Moreover, it is selfadjoint w.r.t.\ 
both inner products and we have
\begin{equation}\label{eqn:relinn} 
(f,g) \,=\, \br f|D^{-1}g\kt\,,\quad f,g\in\ell^2(G)\,.
\end{equation}
The transition matrix $K^G$ for $G$ is defined by \footnote{The transition matrix is also called the
hopping matrix in the physics literature.}
\begin{equation}\label{eqn:defKG}
K^G=AD^{-1}
\end{equation}
and is related to the  
normalized
Laplacian $L^G$ by  
\beq\label{7}
L^G=I-K^G\,,
\eeq
where $I$ is the unit operator on $\ell^2(G)$. It is important to note that $K^G$, and 
hence also $L^G$, is a self-adjoint operator w.r.t.\ the inner product $(\cdot,\cdot)$ as 
a consequence of \eqref{eqn:defKG} and \eqref{eqn:relinn} and the fact that $A$ is 
self-adjoint w.r.t.\ $\br\cdot|\cdot\kt$. Moreover, $L^G$ is non-negative since
\beq
(f,L^Gf) = \sum_{(u,v)\in E(G)} \Big(D^{-1}f(u)-D^{-1}f(v)\Big)^2\,.
\eeq
A central issue in this paper is to determine the spectrum of $L^G$  for a class of graphs 
defined below. Because of the relation \eqref{7} this is equivalent to determining the 
spectrum of $K^G$, which we shall refer to as the spectrum of $G$, and will be our main 
object of study in sections \ref{sec:ACayleytree} and \ref{sec:examples}
\footnote{\parbox[t]{\linewidth} {Note that sometimes a different definition is used; namely, that the spectrum of $A$ is called the spectrum of $G$.}}.

Given an arbitrary graph $G$,  
we next recall the relation between $K^G$ and the simple random walk on $G$. 
By definition, a simple discrete time 
random walk, located at a vertex $u$ at integer time $n$,
moves with equal probabilty, $\sigma_u^{-1}$, to one of the neighbours of $u$ at time $n+1$ .
Let $\omega =(\omega_1,\omega_2,\ldots ,\omega_{n+1})$ be a path from a vertex $u$ to a vertex $v$,
i.e., $\omega_j\in V$, $j=1,\ldots ,n+1$, $\omega_1=u$, $\omega_{n+1} =v$ and
$(\omega_j,\omega_{j+1})\in E$ for $j=1, \ldots n$.  We will refer to $n$ as the length of the path $\omega$
and denote it $|\omega |$.
If the walk is located at $u$ at time $n=0$, then the probability that it is at $v$ after $n$ steps is
given by
\beq\label{rw1}
p_n(u,v)=\sum_{\omega :u\to v} \prod_{j=1}^{n}\sigma^{-1}_{\omega_j}\,,
\eeq
where the sum runs over all paths of length $n$ 
from $u$ to $v$.  If there is no such path then the probability is zero. By convention $p_0(u,u)=1$.  

One checks easily that 
\beq\label{Tn} p_n(u,v)=\br v |(K^G)^n|u\kt\,.
\eeq   
In view of \rf{Tn} one says that $K^G$  generates  simple 
random walk on $G$. 
The generating function for the
probabilities  $p_n(u,v)$ is defined as
\beq\label{rw2}
Q^G_{u,v}(z)=\sum_{n=0}^\infty p_n(u,v) z^n\,,
\eeq
which, using  \eqref{Tn}, is also given by
\beq\label{rw2a}
Q^G_{u,v}(z)=\br v| \left( I-{zK^G}\right)^{-1}|u\kt\,.
\eeq
Note that $\lambda^{-1}Q^G_{u,v}(\lambda^{-1})$ is a matrix element of the 
resolvent $(\lambda - K^G)^{-1}$ of the transition matrix.
 If $u=v$ then we write $Q^G_{u,v}=Q^G_u$, and $Q^G_u(z)$ is the
 generating function for return probabilities to
 $u$.   
 
  For $n>0$, let $p_n^{(1)}(u)$ denote the probability that a walk 
  starting at $u$ is back at $u$ after $n$ steps and that this is the first 
  return to $u$.   The associated generating function is
 \beq\label{rw3}
 P^G_u(z)=\sum_{n=2}^\infty p_n^{(1)}(u) z^n\,,
 \eeq
 called the first return generating function.  We have $P^G_u(1)\leq 1$ 
 since the probabilities $p_n^{(1)}(u)$ 
 refer to mutually exclusive events.  
A walk returning to the starting point $u$ can visit $u$ arbitrarily 
many times before ending at $u$. 
 It follows that
 \beq\label{rw4}
 Q^G_u(z)=\frac{1}{ 1-P^G_u(z)}\,.
 \eeq

 \subsection{The Cayley tree and its spectrum\label{subsec:treeandspectrum}}
 
 A regular Cayley tree of degree $q$ is an infinite tree graph $T_q$ where all the vertices have the same degree $q$.
(This graph is sometimes referred to as the Bethe lattice.)
A planted Cayley tree of degree $q$ is an infinite tree graph $T_q'$ where 
all the vertices have the 
same degree $q$ except one of them, called the root, which has degree 1.  

{The  first return and all-returns generating functions for $T_q$ are independent of the starting vertex, and we denote them by 
$P$ and  $Q$, respectively}.
Moreover, the first return generating  function for $T_q$ obviously equals the first return generating function for $T'_q$ at the root. In order to calculate the latter we note that 
in the first step the walk moves from the root $u$ to its neighbour $v$ with probability $1$.
Clearly, the walk must return to $v$ before taking the final step to 
 $u$ which has probability $q^{-1}$, and it can return to $v$ arbitrarily often before 
 returning to $u$.  Hence, the first return generating function satisfies the equation
 \bea
 P(z) & = & z\sum_{n=0}^\infty \left(\frac{q-1}{ q}P(z)\right)^n  \frac{z}{q}
 \nonumber\\
         &=& \frac{z^2}{q- (q-1)P(z)}\,. \label{app1}
 \eea
 The factor of $z$ is associated with the first step and the factor $z/q$ comes from the last step.
 The solution of \eqref{app1} satisfying the initial condition $P(0)=0$ is
 \beq\label{app2}
 P(z)=\frac{q -\sqrt{q^2-a^2z^2}}{2(q-1)}\,,
 \eeq
 where
 \begin{equation}
     a=2\sqrt{q-1}\,.
 \end{equation}
 By \rf{rw4}  we then find that the return generating function for $T_q$ is
 \beq\label{app3}
 Q(z)=\frac{2-q+ \sqrt{q^2-a^2z^2}}{2(1-z^2)}\,.
 \eeq 

We  can calculate $Q_{u,v}(z)$ for arbitrary vertices $u$ and $v$ in a similar way. First, note that by 
\eqref{rw1} and \eqref{rw2} we have
 \beq\label{app4}
Q_{u,v}(z)= 
\sum_{\omega :u\to v}\left(\frac{z}{q} 
\right)^{|\omega |}\, ,
 \eeq
 where the sum is over all paths $\omega$ from $u$ to $v$.  Let $(u,u_1,\ldots, u_{n-1},v)$ be the 
 shortest path from $u$ to $v$.   Any path 
 $\omega$ from $u$ to $v$ can be uniquely decomposed
 into a sequence of 
 $n+1$ paths $\omega^1,\omega^2, \ldots ,\omega^{n+1}$ where $\omega^1$ is any path from 
 $u$ and back to $u$; the first step in $\omega^2$ is from $u$ to $u_1$, $\omega^2$ avoids $u$
 and returns to $u_1$; the first step in $\omega^3$ is from $u_1$ to $u_2$, $\omega^3$ avoids
 $u_1$ and returns to $u_2$ and so on for the subsequent vertices,
 the first step in $\omega^{n+1}$ being from $u_{n-1}$ to $v$ and otherwise avoiding $u_{n-1}$ before returning to $v$.   Clearly
 \beq\label{app5}
 \sum_{j=1}^{n+1}|\omega^j |=|\omega|\,.
 \eeq
 The sum over $\omega^1$ yields a factor $Q(z)$.  Summing over $\omega^2$ gives
 a factor of
 \beq
\frac{z/q}{1-\frac{q-1}{ q}P(z)}\,,
 \eeq
 where $z/q$ is associated with the initial step in $\omega_2$.  The same factor is obtained 
 by summing over each of the subsequent $\omega^j$'s.  Using \rf{app1} we find that
 \beq\label{app6}
  Q_{u,v}(z)=Q(z)(z^{-1}P(z))^{d(u,v)}\,.
 \eeq
{Note that the same formula applies for a planted Cayley tree $T_q'$ if $u$ is the root, and $v$ any vertex.}

It is important to note that, as a consequence of \eqref{app2} and \eqref{app3}, the functions 
$Q(\lambda^{-1})$ and $\lambda P(\lambda^{-1})$ are both analytic in the complex $\lambda$-plane 
with a cut on the real axis along the interval $[-{\textstyle\frac{a}{q}},{\textstyle\frac{a}{q}}]$, 
and they are real-valued on the real axis outside the cut. 
The limit $\lim_{\epsilon\downarrow 0} Q_u((\lambda +i\epsilon)^{-1})$ exists for all $\lambda\in\mathbb R$ 
and in view of  \eqref{rw2a} this entails (see, e.g.,\ \cite{ReedSimonIV} Sec.\ XIII.6 ) 
that $|u\kt$ belongs to the absolutely continuous subspace 
of $\ell^2(T_q)$ and has spectral density $\rho_u$ concentrated on 
$[-{\textstyle\frac{a}{q}},{\textstyle\frac{a}{q}}]$ given by
\beq\label{eqn:spectral}
\rho_u(\lambda) = \lim_{\epsilon\downarrow 0} \frac{1}{\pi\lambda} \mbox{\rm Im} 
\,Q_u((\lambda+ i\epsilon)^{-1}) = 
\frac{1}{2\pi} \frac{\sqrt{a^2-q^2\lambda^2}} {1-\lambda^2}\,,\quad \lambda\in 
[-{\textstyle\frac{a}{q}},{\textstyle\frac{a}{q}}]\,.
\eeq
Since the vectors $|u\kt$ span all of $\ell^2(T_q)$, it follows that $T_q$ has absolutely 
continuous spectrum equal to $[-{\textstyle\frac{a}{q}},{\textstyle\frac{a}{q}}]$, which  
was first proven in  \cite{spectrum}.  Furthermore,
\begin{equation}\label{eqn:CayleySpecRep}
\bra u f(K)|v\kt = \int_{-{\textstyle\frac{a}{q}}}^{\textstyle\frac{a}{q}} f(\lambda) \,
\rho_{u,v}(\lambda)\, d\lambda 
\end{equation}
for any continuous function $f$ on $[-{\textstyle\frac{a}{q}},{\textstyle\frac{a}{q}}]$, where 
\begin{equation}\label{eq: rhol}
\rho_{u,v}(\lambda)  = \lim_{\epsilon\downarrow 0} \frac{1}{2\pi i
\lambda }\left(Q_{u,v}((\lambda-i\epsilon)^{-1}) -Q_{u,v}((\lambda+i\epsilon)^{-1}) \right)\,.
\end{equation}

Clearly, $\rho_{u,v}$ depends only on the distance $\ell$ between $u$ and $v$ in $T_q$ and hence 
we use the notation 
{$\rho^{(\ell)}$} for this function. Recalling  \eqref{app6} we get from \eqref{eq: rhol} that 
\beq\label{eqn:sphericalfn1}{\rho^{(\ell)}(\lambda)=-\frac{1}{\rho_u(\lambda)} \lim_{\epsilon\downarrow 0}\mathrm{Im} \left( (\lambda+i\epsilon)^{\ell-1}Q((\lambda+i\epsilon)^{-1}) \,P((\lambda+i\epsilon)^{-1})^\ell \right). }\eeq 
Setting $\lambda= aq^{-1}\cos\theta$, $\theta\in[0,\pi]$, and defining $\tan\beta=\frac{q}{q-2}\tan\theta$, gives
\bea \rho^{(\ell)}(\lambda(\theta))
&=&\frac{1}{(q-1)^{\half \ell}}  \frac{\sin(\ell\,\theta+\beta)}{\sin\beta}     \nonumber\\
&=& \frac{1}{q\,(q-1)^{\half \ell}}  \left(2\,C_\ell(\cos\theta) +(q-2)\,U_\ell(\cos\theta)\right)\,,\label{eqn:sphericalfn2}\eea
where $C_\ell$ and $U_\ell$ are Chebyshev polynomials of the first and second kind. (We use $C_\ell$ for the polynomial of the first kind to avoid confusion with our notation for trees.)
 
\section{Asymptotic Cayley Trees}\label{sec:ACayleytree}

\subsection{Definition}\label{subsec:ACayleytree defn}

Given a graph $G$ and a subgraph $F$ we denote by $G\setminus F$ the subgraph of $G$ spanned by 
vertices of $G$ that are not in $F$ 
 (i.e., the edges of $G\setminus F$ are those whose end vertices are both outside $F$). 
 If $T'_q$ is a planted Cayley tree of degree $q$ with root $r$, we set $S_q = T'_q\setminus r$; thus $S_q$ is a tree, 
 {commonly known as the semi-infinite Cayley tree of order $q$,} with all vertices of degree $q$, except one which is of degree $q-1$ and is defined 
 as the root of $S_q$.

 Let $B$ be a graph and let $T'_q$ be a planted tree disjoint from $B$.  A graph defined by 
 identifying the root $r$ of $T'_q$ with some vertex $v_0$ of $B$ is then said to be obtained by 
 grafting $T'_q$ onto $B$ at $v_0$, which we call the graft vertex.

Given a vertex $u_0$ of $G$, let $B_R(u_0)$ denote the subgraph  spanned by the vertices at graph 
distance at most $R$ from $u_0$, i.e. the (closed) ball in $G$ of radius $R$ around $u_0$.
We define an \emph{asymptotic Cayley tree} of degree $q$ to be a connected graph $G$ with the 
property that there exists a vertex $u_0$ and an integer $R\geq 0$ such that $G\setminus B_R(u_0)$ 
is a finite, non-empty union of disjoint trees isomorphic to $S_q$, whose roots are at graph distance 
$R+1$ from $u_0$ in $G$. 
By the triangle inequality it follows that if $G$ satisfies the stated condition for some vertex $u_0$ 
and $R\geq 0$, then it also satisfies the condition with $u_0$ replaced by any vertex $u_1$ and with 
$R$ replaced by any integer larger than or equal to $R+ d_G(u_0,u_1)$.

While the definition just given is formally simple, we shall make frequent use of the following more constructive characterisation of an asymptotic Cayley tree. First, note that if $G$ is obtained by successively grafting planted Cayley trees of degree $q$ onto some finite graph $B$, 
then $G$ is an asymptotic Cayley tree of degree $q$. Indeed, let $u_0$ be some vertex in $B$ and choose 
\beq
R\geq  \mbox{\rm{max}}\,\{d_B(v,v')\mid v,v'\in V(B)\}\,.
\eeq
\noindent Then $G, u_0, R$ satisfy the defining condition of an asymptotic Cayley tree above. Conversely, if $G, u_0, R$ satisfy this condition, then $G$ is obtained by grafting $q-1$ planted Cayley trees onto $B_{R+1}(u_0)$ at each vertex at distance $R+1$ from $u_0$. This proves that asymptotic Cayley trees of degree $q$ are precisely those graphs that can be obtained by grafting a finite (positive) number of planted Cayley trees onto some finite graph $B$.

We call the minimal such finite graph $B$ 
the \emph{core} of $G$ and denote it by $G^{(0)}$. In order to see that $G^{(0)}$ is unique we define a subtree of $G$ to be Cayley-maximal if it is a maximal subtree of $G$ that is isomorphic to $T_q'$ and such that all its vertices except possibly the root have degree $q$ in $G$. It is then easily verified that two Cayley-maximal subtrees of $G$ are either identical or share at most their roots. It follows that removing all such subtrees, save their roots, is a well defined process and yields a finite subgraph $B$ 
that is obviously the smallest one with the desired property. 

 Given an asymptotic Cayley tree $G$, we define $V_0^G\subseteq V(G^{(0)})$ to be the 
set of vertices at which planted Cayley trees are grafted to obtain $G$. Any sequence of the form
\begin{equation}
    {\mathcal G}= \{ G^{(0)},G^{(1)}, G^{(2)}, \ldots , G^{(n)}=G \}\,, \label{eqn:Gsequence}
\end{equation}
where $G^{(j+1)}$ is obtained from $G^{(j)}$ by grafting one or more trees isomorphic to $T'_q$ onto $G^{(j)}$ at a vertex $v_j\in V^G_0$, will be called a $G$-sequence in the following. 

We note that, as the reader may also easily verify, the definition given here of an asymptotic Cayley tree is 
equivalent to the one given in Definition 2.1 of \cite{verdiere} of a graph 
isomorphic to a regular tree at infinity.

\subsection{Generating functions}

We now establish the relationship between the resolvents of the transition matrices on two graphs  that are related by the grafting of planted trees at a single vertex. 
\begin{lem}\label{lemma:resolventmap}
Let $B$ be a (finite or infinite) graph and 
denote by $G$ the graph that is obtained by grafting $p$ copies of $T'_q$ at a vertex  $v_0$, of degree $\sigma_{v_0}$, in $B$, which is then identified with $v_0\in V(G)$. 
Then the following statements relating  $Q^G_{u,v}$ and $Q^B_{u,v}$ hold:
\begin{itemize} 
    \item[(i)] For $u=v=v_0$,
    \begin{eqnarray}
      Q^G_{v_0}(z)
     &=& \frac{(\sigma_{v_0}+p)Q^B_{v_0}(z)Q(z)}{\sigma_{v_0}Q(z) + p\,Q^B_{v_0}(z)}\,,\label{eqn:Qv0map}
\end{eqnarray}
 where $Q^B_{v_0}(z)$ is given by \eqref{rw4}.
    \item[(ii)]For  $u\in V(B)$, $v\in V(B)\setminus\{v_0\}$,
    \begin{eqnarray}
Q^G_{u,v}(z)
&=&Q^{B}_{u,v}(z) -  \frac{p\,Q^{B}_{u,v_0}(z)Q^B_{v_0,v}(z)}{\sigma_{v_0}Q(z) + p\,Q^B_{v_0}(z)}\,.\label{eqn:QG3}
\end{eqnarray}
    \item[(iii)]  For $u\in V(B)$ and $v\in G\setminus B$,  
    \begin{eqnarray}
Q^G_{u,v}(z)
&=& \frac{q\,Q(z)Q^B_{u,v_0}(z)}{\sigma_{v_0}Q(z) + p\,Q^B_{v_0}(z)} \left(z^{-1}P(z)\right)^{d_G(v_0,v)}\,.\label{eqn:QG2}
\end{eqnarray}
    \item[(iv)] For $u,v\in G\setminus B$,
     \begin{eqnarray} Q^{G}_{u,v}(z) &=&   Q(z) \,\left(1-\alpha_p(z)\, \left( z^{-1}P(z)\right)^{d_G(v_0,u)+d_G(v_0,v)-d_G(u,v)}\right)\nonumber\\
     &&\qquad\times \left( z^{-1}P(z)\right)^{d_G(u,v)}\,,\label{eqn:QG1}
     \end{eqnarray}
where
\beq\label{eqn:alphap}
\alpha_p(z)=1-\frac{q\,Q^B_{v_0}(z)}{\sigma_{v_0}Q(z)+pQ^B_{v_0}(z)}\,.
\eeq
\end{itemize}
\end{lem}
\begin{rem}
Note that, since 
\begin{equation}
    \sigma_u Q^G_{u,v}(z)=\sigma_v Q^G_{v,u}(z)\, \label{eqn:Qidentity}
\end{equation} 
holds for any graph $G$,  the lemma is sufficient to determine $Q^G_{u,v}(z), \forall u,v\in V(G)$ given 
$Q^B_{u,v}(z), \forall u,v\in V(B)$.
\end{rem}
\begin{proof}
First note that walks starting from and returning to $v_0$ consist of repeated excursions into either $B$ or a tree which gives
\begin{eqnarray}
     Q^G_{v_0}(z)&=&\frac{1}{1-p\frac{P(z)}{\sigma_{v_0}+p}-\frac{\sigma_{v_0}P^B_{v_0}(z)}{\sigma_{v_0}+p}}\,,
     \end{eqnarray}
     where $P^B_{v_0}(z)$ is given by \eqref{rw3}.
     Statement (\emph{i})
           follows by \eqref{rw4}.

To prove  statement (\emph{ii}) we first need the identity
\begin{eqnarray}
     Q^G_{v_0,v}(z)&=&Q^G_{v_0}(z)\frac{\sigma_{v_0}}{\sigma_{v_0}+p}\frac{Q^B_{v_0,v}(z)}{Q^B_{v_0}(z)}\nonumber\\
     &=& \frac{\sigma_{v_0}Q^B_{v_0,v}(z)Q(z)}{\sigma_{v_0}Q(z) + p\,Q^B_{v_0}(z)}\,,\quad v\in V(B)\setminus\{v_0\}\,,\label{eqn:Quseful}
\end{eqnarray}
 which is easily obtained by decomposing a walk contributing to the left-hand side into a (possibly trivial) 
 walk from $v_0$ and back to $v_0$ and a walk from $v_0$ to $u$ that does not return to $v_0$, and then using \eqref{eqn:Qv0map}. The statement for $u=v_0$ follows by rearranging \eqref{eqn:Quseful}.
 For $u\ne v_0$
separate the contributing walks into those that do, or do not, visit $v_0$ gives
\begin{eqnarray}
Q^G_{u,v}(z)&=&\frac{Q^B_{u,v_0}(z)Q^G_{v_0,v}(z)}{Q^B_{v_0}(z)}+\left( Q^{B}_{u,v}(z) -\frac{Q^B_{u,v_0}(z)Q^B_{v_0,v}(z)}{Q^B_{v_0}(z)}\right)\nonumber \\
&=&Q^{B}_{u,v}(z) -  \frac{p\,Q^{B}_{u,v_0}(z)Q^B_{v_0,v}(z)}{\sigma_{v_0}Q(z) + p\,Q^B_{v_0}(z)}\,,\end{eqnarray}
where we have used \eqref{eqn:Quseful} in the second step. This completes the proof of statement (\emph{ii}).

Next consider the case when   $u\in V(B)$ and $v\in G\setminus B$. Decompose the contributing walks  into a component that starts at $u$ and ends at $v_0$, followed by a walk from $v_0$ to $v$ without revisiting $v_0$.  Applying the same reasoning that led to \rf{app6}, but noting that $v_0\in G$ has degree $\sigma_{v_0}+p$ rather than  $q$, we obtain
\begin{eqnarray}
Q^G_{u,v}(z)&=& Q^G_{u,v_0}(z) \frac{q}{\sigma_{v_0}+p} \left(z^{-1}P(z)\right)^{d_G(v_0,v)}\,.
\end{eqnarray}
Using \eqref{eqn:Qv0map} if $u=v_0$, respectively \eqref{eqn:Qidentity} and \eqref{eqn:Quseful} if $u\ne v_0$, we obtain \eqref{eqn:QG2} which proves statement (\emph{iii}). 

Denote by $\Bbar$ the subgraph of $G$ obtained by removing all of $B$ except the vertex $v_0$, i.e.  $\Bbar=G\setminus \{B\setminus\{v_0\}\}$.
Then, for $u,v\in G\setminus B$,
separate the contributing walks into those that do, or do not, visit $v_0$ to obtain
\begin{eqnarray}
Q^G_{u,v}(z)&=&\frac{Q^\Bbar_{u,v_0}(z)Q^G_{v_0,v}(z)}{Q^\Bbar_{v_0}(z)}+\left( Q^{\Bbar}_{u,v}(z) -\frac{Q^\Bbar_{u,v_0}(z)Q^\Bbar_{v_0,v}(z)}{Q^\Bbar_{v_0}(z)}\right)\,. \label{eqn:work0}
\end{eqnarray}
To evaluate the first term 
note that, by \eqref{eqn:Qidentity} and then \eqref{app6}, 
\begin{eqnarray}
\frac{Q^\Bbar_{u,v_0}(z)}{Q^\Bbar_{v_0}(z)}=\frac{p}{q}\frac{Q^\Bbar_{v_0,u}(z)}{Q^\Bbar_{v_0}(z)}=\frac{p}{q}\frac{\frac{q}{p}Q_{v_0,u}(z)}{Q(z)}=\left(z^{-1}P(z)\right)^{d_G(v_0,u)}.
\end{eqnarray}
Combining this with  \eqref{eqn:QG2} gives
\begin{eqnarray}
 \frac{q\,Q(z)Q^B_{v_0}(z)}{\sigma_{v_0}Q(z) + p\,Q^B_{v_0}(z)} \left(z^{-1}P(z)\right)^{d_G(v_0,v)+d_G(v_0,u)}\,.\label{eqn:work1}
\end{eqnarray}
The second term in \eqref{eqn:work0}
occurs only when $u,v$ are in the same planted tree $T_q'$. Since the contributing walks to the term do not reach $v_0$, we can compute it by replacing $\Bbar$ with the regular Cayley tree obtained by identifying the root of $T_q'$ with the root of the semi-infinite Cayley tree $S_q$.  Using \eqref{app6} we then obtain
\begin{eqnarray}
    Q(z)\left(\left(z^{-1}P(z)\right)^{d_G(u,v)}-\left(z^{-1}P(z)\right)^{d_G(v_0,v)+d_G(v_0,u)}\right)\,.
\end{eqnarray}
Combining the two contributions proves statement (\emph{iv}).
\end{proof}

\subsection{The spectrum of asymptotic Cayley trees}

{Here we study 
the spectrum of an asymptotic Cayley tree $G$.  
Our strategy is to deploy Lemma \ref{lemma:resolventmap} on 
the sequence $\mathcal G$ \eqref{eqn:Gsequence}. 
By grafting trees to one vertex at a time,  we can follow the evolution of 
the spectrum from one member of the sequence to its successor.}

\begin{thm} \label{thm:continuous}
The spectrum of an asymptotic Cayley tree $G$ of degree $q$ consists  of an absolutely 
continuous part  $[-\frac{a}{q}, \frac{a}{q}]$ and a finite set of eigenvalues $S^G_{pp}$ (the 
pure point  spectrum).
\end{thm}
\begin{proof}
Let $\mathcal G$ be any $G$-sequence as in \eqref{eqn:Gsequence} 
and assume for some index  $i$ and arbitrary vertices $u,v$ in $G^{(i)}$ that $Q^{G^{(i)}}_{u,v}$ has the form
\begin{equation}\label{eqn:Qform}
    Q^{G^{(i)}}_{u,v}(z)= A^{G^{(i)}}_{u,v}(z) + \tilde A^{G^{(i)}}_{u,v}(z)\,\sqrt{q^2-a^2z^2}\,,
\end{equation}
where $A^{G^{(i)}}_{u,v}$ and ${\tilde A}^{G^{(i)}}_{u,v}$ are rational functions of $z$. Since $Q(z)$ and $z^{-1}P(z)$ are of this form it follows from Lemma \ref{lemma:resolventmap}  applied to $B=G^{(i)}$, that the same holds for $Q^{G^{(i+1)}}_{u,v}(z)$. Recalling that $G^{(0)}$ is a finite graph, we have that $Q^{G^{(0)}}_{u,v}$ is a rational function and it follows by repeating the argument that \eqref{eqn:Qform} also holds for $G^{(n)} =G$. 
 
 Likewise, it follows from the identities of Lemma \ref{lemma:resolventmap} that if $Q^{G^{(i)}}_{u,v}$ is finite outside a finite set 
 {$F_i\subset \mathbb C$}, independent of $u,v$, then $Q^{G^{(i+1)}}_{u,v}(z)$ is finite if $z$ is not in $F_i$ and  $p Q^{G^{(i)}}_{v_0}(z) + \sigma_{v_0} Q(z)\neq 0$, where $v_0$ is the graft vertex in $G^{(i)}$ and $p$ is the number of trees grafted at $v_0$. In view of \eqref{eqn:Qform}, the equation \beq
 p Q^{G^{(i)}}_{v_0}(z) + \sigma_{v_0} Q(z)=0
 \eeq
 has finitely many solutions. Thus $Q^{G^{(i+1)}}_{u,v}(z)$ is finite outside a finite set $F_{i+1}$. By repeating this argument we conclude that the poles of the functions $A^G_{u,v}$ and  $\tilde A^G_{u,v}$ are contained in a finite set $F$ independent of $u,v$. In particular, it follows that the limit $\lim_{\epsilon\downarrow 0} (\lambda + i\epsilon)^{-1}Q^G_{u}((\lambda + i\epsilon)^{-1})$ exists outside $F$ for all vertices $u$, which rules out the presence of a singular continuous spectrum (see, e.g., the Proposition in Section XIII.6 of \cite{ReedSimonIV}). The remaining singularities of $\lambda^{-1}Q^G_{v_0}(\lambda^{-1})$ are then  poles contained in $F$. Thus we have that 
 \begin{equation}\label{eqn:Q_uform}
    \lambda^{-1}Q^G_u(\lambda^{-1})= B^G_{u}(\lambda) + \tilde B^G_{u}(\lambda)\,\sqrt{q^2\lambda^2-a^2}\,,
\end{equation}
 where $B^G_u$ and $\tilde B^G_u$ are rational functions of $\lambda$ that are real-valued on 
 the real axis (by the self-adjointness of $K^G$) and $\tilde B^G_u$ has no poles 
 in $[-\frac{a}{q},\frac{a}{q}]$, which therefore equals the continuous spectrum 
 with the spectral density of the state $|u\kt$ given by
 \beq
 \rho_u(\lambda) = - \frac{1}{\pi} \tilde B^G_u(\lambda)\sqrt{a^2-q^2\lambda^2}\,.
 \eeq
 This completes the proof.
\end{proof}
 The characterization of the spectrum given by Theorem \ref{thm:continuous} and 
 its proof implies, by the spectral theorem, that for $\lambda\notin S^G_{pp}$ we have 
\begin{equation}\label{eqn:specrepQ}
\lambda^{-1}Q^G_{u,v}(\lambda^{-1}) = \sum_{\mu\in S^G_{pp}} 
\frac{\br u|e_{\mu}|v\kt}{\lambda-\mu} -\frac{1}{\pi} \int_{-\frac aq}^{\frac aq}
\frac{{\tilde B}^G_{u,v}(\mu)\sqrt{a^2-q^2\mu^2}}{\lambda - \mu}\,d\mu\,,    
\end{equation}
where $e_\mu$ denotes the spectral projection onto the eigenspace of $K^G$ corresponding to $\mu$
and ${\tilde B}^G_{u,v}(\mu)$ is a rational function of $\mu$ with no poles in  $\cut$.  
For an arbitrary, say continuous, function $f$ on the spectrum of $G$ the 
 following generalisation of \eqref{eqn:CayleySpecRep} holds:
\begin{equation}\label{eqn:specrepgen}
\br u|f(K^G)|v\kt = \sum_{\mu\in S^G_{pp}} f(\mu) \br u |e_{\mu}|v\kt  - 
\frac{1}{\pi}\int_{-\frac aq}^{\frac aq} f(\mu) {\tilde B}^G_{u,v}(\mu)\sqrt{a^2-q^2\mu^2}
\,d\mu\,.     
\end{equation}

\medskip

We next establish some basic estimates for the multiplicity of the individual eigenvalues of $K^G$. 
Denote by ${\hat S}^G_0$ the set of eigenvalues of 
$K^{G^{(0)}}$ that have at least one non-trivial eigenfunction vanishing on the set $V^G_0$ of 
graft vertices in $G^{(0)}$. Moreover, 
denote the subspace of such eigenfunctions with eigenvalue $\lambda$ by $E^G_0(\lambda)$, and 
its dimension by $d^G_0(\lambda)$.

\begin{thm}\label{thm:decoupled}
Let $G$ be an asymptotic Cayley tree of degree $q$. All eigenvalues $\lambda \in S^G_{pp}$ have finite multiplicity and fulfill:
\begin{itemize}
\item[(i)] If $\lambda \in [-\frac{a}{q}, \frac{a}{q}]$, then $\lambda\in \hat S^G_0$ and the multiplicity 
of $\lambda$ equals $d^G_0(\lambda)$;
\item[(ii)] If $\lambda \notin [-\frac{a}{q}, \frac{a}{q}]$, then the multiplicity of $\lambda$ is at most $\sharp V^G_0 + d^G_0(\lambda)$, where $\sharp V$ denotes the number of vertices in $V$.
\end{itemize}
\end{thm}
\noindent For the proof of this result we need the following lemma.
\begin{lem}\label{lemma:sphericalLe}
Let $G$ be an asymptotic Cayley tree with a subtree $T_q'$ grafted at vertex $v_0$, and assume $\phi\in \ell^2(G)$ 
is an eigenfunction of $K^G$ with eigenvalue $\lambda$. Then the following statements hold:
\begin{itemize}
    \item[(i)] If $\lambda\in \cut $, then $\phi$ vanishes on $T'_q$;
    \item[(ii)] If $\lambda\notin \cut $, then $\phi$ is spherical in $T'_q$, i.e.,
    its value at $v\in T'_q$ depends only on the height $d(v_0,v)$,  and  is given by  
    \begin{align}
         \phi(v) = A\, \xi^{-d(v_0,v)}\,,\quad v\in T'_q\,,\label{eqn:expdecay}
    \end{align}
    where $\xi^{-1}= \lambda P(\lambda^{-1}) $ 
    and $A=\phi(v_0)$.  
\end{itemize}
\end{lem}
\begin{proof} Denoting by $v_1$ the unique neighbour of $v_0$ in $T'_q$, we define $\psi (v)=\phi (v)$, if $v\notin T'_q$ or if $v=v_0$ or if $v=v_1$, and
\beq\label{le1}
\psi (v)=\frac{1}{ (q-1)^{h-1}}\sum_{w\in T'_q:d(v_0,w)=h} \phi(w)
\eeq
for vertices $v$ at height $h>1$ in  $T_q'$. Then   $\psi$ is a function of $h=d(v_0,v)$ alone for $v\in T'_q$,
and it is easily seen that it fulfils 
\beq\label{eigeq}
K^G\psi (v) = \lambda \psi(v)\,, \quad v\in G\,.
\eeq
Moreover, $\psi$ belongs to $\ell^2(G)$ as a consequence of the following estimate, 
where $\phi_0$ denotes the restriction of $\phi$ to $V(G)\setminus V(T'_q)$:
\begin{eqnarray}
    \|\psi\|_s^2 & = & \| \phi_0\|_s^2 + |\psi(v_0)|^2 +\sum_{h=1}^\infty (q-1)^{-h+1}\Big|
    \sum_{w\in T'_q:d(v_0,w)=h} \phi(w)\Big|^2\nonumber\\
&\leq & \| \phi_0\|_s^2 + |\phi(v_0)|^2 +\sum_{h=1}^\infty \sum_{w\in T'_q:d(v_0,w)=h} |\phi(w)|^2
= \|\phi \|_s^2 \,,\label{le2}
\end{eqnarray}
where we have used  the Cauchy-Schwarz inequality. Hence, $\psi$ is an eigenfunction of $K^G$ with eigenvalue $\lambda$.
 
 Writing $\psi(v)=\psi(h)$ for $v\in T'_q$ at height $h$, the eigenvalue condition \eqref{eigeq} gives 
\beq\label{le3}
(q-1)\psi (h+1)+\psi (h-1)=\lambda q\psi (h)\,,\quad h\geq 1\,.
\eeq
 The general solution to this equation is
\beq\label{le4}
\psi (h)=A\xi_+^{-h}+ B\xi_-^{-h}\,,\quad h\geq 0\,,
\eeq
where $A$ and $B$ are constants and $\xi_\pm$ are the solutions to the quadratic equation 
\beq\label{le5}
\xi^2 -\lambda q\xi +q-1=0\,.
\eeq
For $|\lambda |<a/q$ we find
\beq\label{le6}
\xi_\pm = \sqrt{(q-1)} e^{\pm i\theta},
\eeq
where $\theta$ is real, so $\psi$ cannot be in $\ell^2(G)$ unless it vanishes on $T'_q$. In particular, 
$\phi(v_0)=\psi(v_0)=0$ and $\phi(v_1)=\psi(v_1)=0$. It follows that if we consider $q$ copies of $T'_q$ with 
the restriction of $\phi$ defined on each of them, then by identifying their roots we obtain an eigenfunction 
on the Cayley tree $T_q$ with eigenvalue $\lambda$. But since the pure point spectrum of $T_q$ is empty, it 
follows that $\phi$ vanishes on $T'_q$, thus proving statement  $(i)$ in this case.  A similar argument applies if $\lambda= \pm\frac{a}{q}$, in which case the general solution to \eqref{le3} is $\psi(h) = (A + B\,h)(q-1)^{-\frac h2}$. This completes the proof of 
statement $(i)$.

If $|\lambda |>\frac aq$, then $\xi_{+}\neq \xi_{-}$ are real and $\xi_+\xi_- = q-1$. Hence, exactly 
one of the roots, say $\xi_+$, fulfills $|\xi_+|>\sqrt{q-1}$ and the only solutions in $\ell^2(T'_q)$ 
are of the form $A\,\xi_{+}^{-h}, h\geq 0$. In particular, we have \beq\label{eq:rat}
\phi(v_0) = \psi(v_0) = \xi_+\psi(v_1) = \xi_+\phi(v_1)\,.
\eeq
Applying this argument to any subtree of $T'_q$ spanned by the descendants of any vertex $w_1\neq v_0$ in $T'_q$ 
and its predecessor $w_0$ it follows that $\phi(w_0)=\xi_+\phi(w_1)$. This proves that the restriction of 
$\phi$ to $T'_q$  is of the form $A\,\xi_+^{-d(v_0,v)}$ for $v\in T'_q$. 
Finally, it is straighforward to check that $\xi_+^{-1}=\lambda P(\lambda^{-1})$ by  direct substitution 
in \eqref{le5} and applying \eqref{app1}. 
This proves statement $(ii)$.
\end{proof}

\begin{proof}[Proof of Theorem \ref{thm:decoupled}]
Let $\phi\in \ell^2(G)$ be an eigenfunction of $K^G$ with eigenvalue $\lambda\in[-\frac aq,\frac aq]$. By Lemma 
\ref{lemma:sphericalLe}$(i)$ the function $\phi$ vanishes on all grafted trees $T'_q$, from which it follows 
that the restriction $\phi^{(0)}$ of $\phi$ to $G^{(0)}$ belongs to $E^G_0(\lambda)$. On the other hand, 
extending any $\phi^{(0)}\in E^G_0(\lambda)$ to $V(G)$ by assigning the value $0$ at vertices outside  $G^{(0)}$ we 
obtain an eigenfunction of $K^G$ with eigenvalue $\lambda$. This proves statement $(i)$.

Let $\lambda\in S^G_{pp}\backslash [-\frac aq,\frac aq]$ and denote the graft vertices in $G^{(0)}$ by $v_1,\dots, v_N$, 
where $N=\sharp V^G_0$.  Moreover, let $\phi_i, 1\leq i\leq N,$ denote an eigenfunction of $K^G$ with eigenvalue $\lambda$ 
such that 
\begin{equation}
\phi_i(v_i)\neq 0\quad \mbox{and} \quad \phi_i(v_j) = 0\quad \mbox{for $j<i$}\,,
\end{equation}
if such a function exists. Otherwise, set $\phi_i =0$. By Lemma \ref{lemma:sphericalLe}$(ii)$ the restriction 
of $\phi_i$ to the trees $T'_q$ grafted at $v_i$ is uniquely determined up to multiplication by a constant and 
it vanishes on the trees grafted at any vertex $v_j,\, j<i$. It follows that for any eigenfunction $\phi$ of $K^G$ 
with eigenvalue $\lambda$ there exist coefficients $\mu_1,\dots,\mu_N$ such that $\phi-\mu_1\phi_1-\dots-\mu_N\phi_N$ 
vanishes on all grafted trees $T'_q$. By the same arguments as above it follows that the dimension of the space of 
such functions equals $d^G_0(\lambda)$. Clearly this proves statement $(ii)$.  
\end{proof}

\medskip

Our final result on the spectrum of $G$ provides bounds on the total dimension of eigenspaces of $G$ 
in terms of the same quantity for the core of $G$.
\begin{thm}\label{thm:point}
Let $G$ be an asymptotic Cayley tree and let $d^G(\lambda)$ denote the multiplicity of the eigenvalue 
$\lambda\in S^G_{pp}$, i.e., the dimension of the eigenspace of $K^G$ corresponding to $\lambda$. 
Then the following relations hold: 
\beq
\sum_{\mu\in S^{G^{(0)}}_{pp}\setminus  [-\tfrac{a}{q},\tfrac{a}{q}]} d^{G^{(0)}}(\mu) - 2\,\sharp V^G_0\; \leq  \;\sum_{\mu\in S^G_{pp}\setminus  [-\tfrac{a}{q},\tfrac{a}{q}]} d^G(\mu)\;\leq \;\sum_{\mu\in S^{G^{(0)}}_{pp}\setminus  [-\tfrac{a}{q},\tfrac{a}{q}]} d^{G^{(0)}}(\mu)\,.
\eeq
\end{thm}

\noindent Before proving this result it is convenient to establish the following lemma. 
\begin{lem}\label{lemma:Qpolesinz}
Let $G$ be an asymptotic Cayley tree obtained by grafting $p\geq 1$ trees isomorphic to $T'_q$ at a vertex 
$v_0$ of a graph $B$ (which is either finite or an asymptotic Cayley tree), and assume $w_k, k=1\ldots m$, and $z_j, j=1\ldots n$, satisfying 
$w_m<w_{m-1}<\dots < w_1< 0< z_1<z_2<\dots<z_n$ are the poles of $Q^B_{v_0}$ in $(-\frac qa,\frac qa)$. 
Then $\alpha_p$, given by \eqref{eqn:alphap}, has exactly one pole $z^*_i\in (z_i,z_{i+1})$ for each 
$i=1,\dots, n-1,$ and $w_j^*\in (w_{j+1},w_j)$ for each $j=1,\dots,m-1$. Moreover, $\alpha_p$ has no pole in $(w_1,z_1)$ 
and at most one pole in each of the intervals $(-\frac qa,w_m)$ and $(z_n,\frac qa)$. 
\end{lem}

\begin{proof}
By \eqref{eqn:specrepQ} 
the poles of $Q^B_{v_0}(z)$ are of the form $\lambda^{-1}$, where $\lambda\in S^B_{pp}\setminus\{0\}$ has an 
eigenfunction $\phi_\lambda$ that overlaps $|v_0\kt$, i.e., $\phi_\lambda(v_0)\neq 0$. By Theorem 
\ref{thm:decoupled}$(i)$ this implies that $\lambda\notin [-\frac aq,\frac aq]$ and so the eigenvalues in 
question are $z_i^{-1},\, i=1,\dots, n,$ and $w_j^{-1},\,j=1,\dots, m$, and \eqref{eqn:specrepQ} yields 
\begin{eqnarray}
zQ_{v_0}^B (z) 
 & = & \int_{-a/q}^{a/q}  \frac{z\rho^B_{v_0} (\lambda )}{1-z\lambda}\,d\lambda\nonumber \\ 
 & + &\sum_{i=1}^n \frac{za_i}{ 1-zz_i^{-1}} + \sum_{j=1}^m \frac{zb_j}{ 1-zw_j^{-1}}\,, 
  \label{135}
 \end{eqnarray}
where the coefficients $a_i$ and $b_j$ are positive numbers given by
\beq
a_i=\br v_0 |e_i| v_0\kt \quad\mbox{and}\quad b_j=\br v_0|f_j|v_0\kt \,.
\eeq
Here $e_i$ and $f_j$ denote the spectral projections of $K^B$ in $\ell^2(B)$  corresponding to 
$z_i^{-1}$ 
and $w_j^{-1}$, respectively. It follows by differentiation that $zQ^B_{v_0}(z)$ is a 
strictly increasing function 
on each of the open subintervals of $(-\frac qa,\frac qa)$ arising from the subdivision $-\frac qa < w_m <w _{m-1}<\dots 
< w_1 < 0 < z_1 < z_2 <\dots < z_n < \frac qa$. 

 Since $Q(z)\neq 0$ for $z\in (-\frac qa,\frac qa)$ we see from \rf{eqn:alphap} that $\alpha_p$ has poles at values 
 of $z$ for which 
 \beq\label{136}
 p\,Q^B_{v_0}(z)= - \sigma_{v_0}Q(z)\,,
 \eeq
and is otherwise analytic on $(-\frac qa,\frac qa)$. The function $-zQ(z)$ \eqref{app3} takes the value $0$ at $z=0$ 
and is decreasing for $0<z<q/a$, and the function $zQ^B_{v_0}(z)$ vanishes at $z=0$ and is increasing on the 
interval $(0,z_1)$, so in this interval there is no
 solution to \rf{136}.  In the intervals $(z_i,z_{i+1})$ the function
 $zQ^B_{v_0}(z)$ increases monotonically from $-\infty$ at $z=z_i$ to $+\infty$ at $z_{i+1}$, so
 in each of these intervals there is a unique solution, $z_i^*$,   to \rf{136}, see 
 Fig.\ \ref{fig:Qbranches}. 
 \begin{figure}[ht]
\centering
    \includegraphics[scale=0.25,trim={10cm  1.5cm  2cm 0cm},clip]{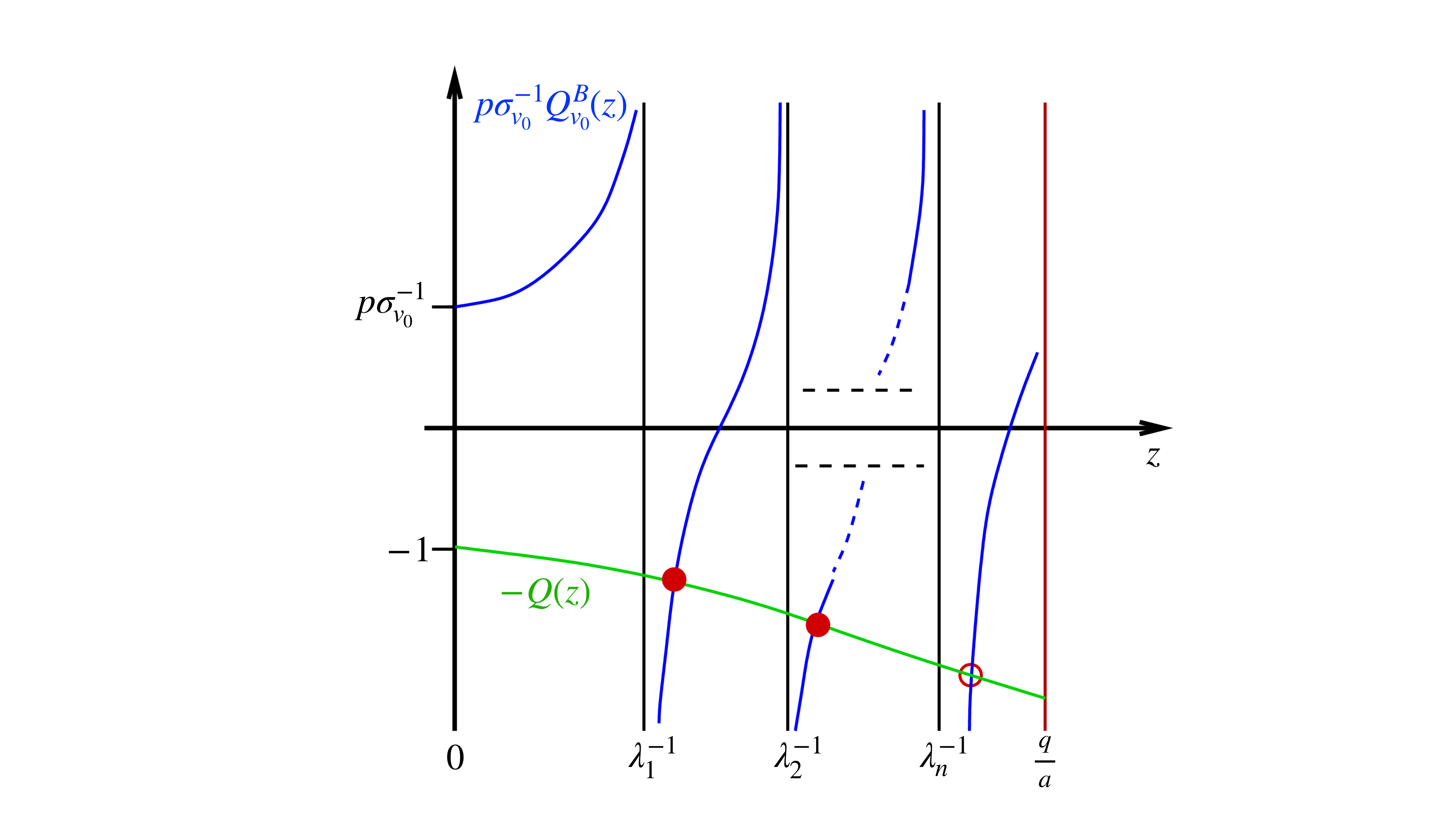}
    \caption{The solutions to the pole condition \eqref{136}. Blue (dark grey) lines represent the branches of  $ p\sigma_{v_0}^{-1}Q^{B}_{v_0}(z)$,
   and the green (light grey) line is $-Q(z)$. Solid points indicate solutions $z^*$ that must exist; the solution indicated by the open  point exists only if the condition 
   \eqref{137} is satisfied. }
   \label{fig:Qbranches} 
\end{figure} 
\noindent In the interval
 $(z_n, q/a)$ there is a solution $z_n^*$ if and only if
 \beq\label{137}
 \sigma^{-1}_{v_0}p\,Q^B_{v_0}\left(\frac qa\right)> - Q\left(\frac qa\right) = -\frac{2q-2}{q-2}\,.
 \eeq
 This proves the statement concerning poles in $(0,\frac aq)$. Similar arguments apply for the poles in $(-\frac aq,0)$, and clearly there is no pole at $z=0$. 
 Hence, the proof is complete.
 \end{proof}

\begin{proof}[Proof of Theorem \ref{thm:point}] Let $\cal G$ be a $G$-sequence as in \eqref{eqn:Gsequence} 
with $n=\sharp V^G_0$, i.e., such that all trees $T'_q$ to be grafted at any vertex in $V^G_0$ are grafted in one step. 
Using the notation of Lemma \ref{lemma:Qpolesinz} with $B=G^{(i)}$ and with $G$ denoting $G^{(i+1)}$, it then 
suffices to prove that 
\beq
\sum_{\mu\in S^B_{pp}\setminus  [-\tfrac{a}{q},\tfrac{a}{q}]} d^B(\mu) - 2 \; \leq \;\sum_{\lambda\in S^G_{pp}\setminus  [-\tfrac{a}{q},\tfrac{a}{q}]} d^G(\lambda)\;  \leq 
\;\sum_{\mu\in S^B_{pp}\setminus  [-\tfrac{a}{q},\tfrac{a}{q}]} d^B(\mu)\,.
\eeq
Consider first $Q^G_u(z)$, where $u\notin B$. Recalling again that $Q(z)$ and $z^{-1}P(z)$ are analytic in 
$(-\frac qa,\frac qa)$, it follows from \eqref{eqn:QG1} that the poles of $Q^G_u(z)$ are precisely those of 
$\alpha_p(z)$ as given by Lemma \ref{lemma:Qpolesinz} and hence their inverse values, $(z^*_i)^{-1}$ and $(w^*_i)^{-1}$, 
belong to $S^G_{pp}$. 

In order to compare eigenvalue multiplicities for $B$ and $G$ we consider three possible cases for an 
eigenvalue $\lambda$ of $K^G$ outside $[-\tfrac{a}{q}, \tfrac{a}{q}]$:

\smallskip

1)\; $\lambda^{-1}$ is different from all $z_i, z^*_i, w_j, w^*_j$. 
{By Lemma \ref{lemma:resolventmap}(\emph{i}), 
$\lambda^{-1}$ is not a pole of $Q^G_{v_0}$.} Hence all corresponding eigenfunctions of $K^G$ vanish at $v_0$ and are extensions of eigenfunctions of $K^B$ with eigenvalue $\lambda$. Since also $Q^B_{v_0}$ does not have a pole at 
$\lambda^{-1}$ it follows that 
\beq
d^G(\lambda) = d^B(\lambda)\,.
\eeq

2)\; $\lambda= z_i^{-1}$. In this case, $K^B$ has an eigenfunction $\phi_\lambda$ with eigenvalue $\lambda$ that overlaps $|v_0\kt$. Hence, the eigenspace of $K^B$ has a basis consisting of $\phi_\lambda$ and $d^B(\lambda)-1$ eigenfunctions that vanish at $v_0$. The latter can be extended trivially to eigenfunctions of $K^G$ with eigenvalue $\lambda$. On the other hand, $K^G$ can have no such eigenfunction overlapping $|v_0\kt$, since otherwise $Q^G_{v_0}$ would have a pole at $z_i$, which is not the case by 
{Lemma \ref{lemma:resolventmap}(\emph{i})}. 
We conclude that 
\beq
d^G(\lambda) = d^B(\lambda) - 1\,.
\eeq
Clearly, the same statement holds if $\lambda = w_j^{-1}$.  

3)\;  $\lambda = (z^*_i)^{-1}$. Since $Q^G_{v_0}$ has a pole at $z^*_i$ by 
{Lemma \ref{lemma:resolventmap}(\emph{iv})}  
we see that 
$K^G$ has an eigenfunction $\psi_\lambda$ with eigenvalue $\lambda$ that overlaps $|v_0\kt$ and hence the 
corresponding eigenspace of $K^G$ has a basis consisting of $\psi_\lambda$ and $d^G(\lambda)-1$ 
eigenfunctions that vanish at $v_0$ (and  on the trees grafted at $v_0$ by Lemma \ref{lemma:sphericalLe}). 
Restricting the latter to $B$ we obtain $d^G(\lambda) -1$ linearly independent eigenfunctions of $K^B$ with 
eigenvalue $\lambda$. On the other hand, $K^B$ does not have any such eigenfunctions that overlap 
$|v_0\kt$ since $Q^B$ is regular at $z_i^*$. Hence,
\beq
d^G(\lambda) = d^B(\lambda) +1
\eeq
in this case. The same conclusion holds for $\lambda= (w_j^*)^{-1}$. 

\smallskip

Note that if $\lambda\notin [-\tfrac{a}{q},\tfrac{a}{q}]$ 
{is an eigenvalue of $K^B$ that is not among the $z_i$ and $w_j$ 
then $Q_u(z)$ has a pole at $\lambda^{-1}$ for some $u\in V(B)$ and} 
{Lemma \ref{lemma:resolventmap}(\emph{ii})} 
then implies that $\lambda$ is an eigenvalue of $K^G$. According to 1) above the 
contribution of such eigenvalues to the sum of all eigenvalue multiplicities is the same for $K^B$ and $K^G$. 
For the remaining eigenvalues we get, using 2) and 3) above, 
\begin{align}
 \sum_{i=1}^n d^G(z_i^{-1}) + \sum_{j=1}^m d^G(w_i^{-1})
= & \sum_{i=1}^n (d^B(z_i^{-1}) -1) + \sum_{j=1}^m (d^B(w_j^{-1})-1)\nonumber\\
= & \sum_{i=1}^n d^B(z_i^{-1}) + \sum_{j=1}^m d^B(w_j^{-1}) -(n+m)
\end{align}
and 
\begin{align}
\sum_{i} d^G((z^*_i)^{-1}) + \sum_{j} d^G((w^*_i)^{-1})
= & \sum_{i} (d^B((z^*_i)^{-1}) +1) + \sum_{j} (d^B((w^*_j)^{-1})+1)\nonumber\\
= & \sum_{i} d^B((z^*_i)^{-1}) + \sum_{j} d^B((w^*_j)^{-1}) +(k+l)\,,
\end{align}
where $k=n-1$ or $k=n$ and $l=m-1$ or $l=m$ by Lemma \ref{lemma:Qpolesinz}. Finally, collecting the contributions to the sums of all eigenvalue multiplicities for $K^G$ and $K^B$ so displayed, corresponding to eigenvalues outside  $[-\tfrac{a}{q},\tfrac{a}{q}]$, the claimed inequalities follow.
\end{proof}
\begin{rem} Our results have a number of implications:
\begin{enumerate}
    \item Combining Theorem \ref{thm:decoupled}(\emph{ii}) and the upper bound of Theorem \ref{thm:point} we also have 
    \begin{eqnarray}
        \sum_{\mu\in S^G_{pp}} d^G(\mu)\;\leq \;\sum_{\mu\in S^{G^{(0)}}_{pp}} d^{G^{(0)}}(\mu)=\sharp V(G^{(0)})\,,
    \end{eqnarray}
    which is the analogue of the bound given in \cite{verdiere},\ Theorem 4.2, for the adjacency matrix.
    \item The lower bound of Theorem \ref{thm:point} 
    shows that the presence of $\ell^2$ states is ubiquitous. Specifically, if the number of eigenfunctions of $K^{G^{(0)}}$ with eigenvalues lying outside the continuous spectrum  exceeds twice the number of graft vertices, then $G$ necessarily has a non-empty pure point spectrum.
    \item From the proof of Lemma \ref{lemma:Qpolesinz} we see that, as trees are added to the core graph to generate the sequence $\mathcal{G} $, discrete eigenvalues either remain unchanged or move towards the continuum.  
        \end{enumerate}
\end{rem}

\section{Examples}\label{sec:examples}
In this section we consider simple examples of asymptotic Cayley trees where the pure point spectrum 
can be calculated rather explicitly.   First we look at the case when the core is a complete graph
and we graft the same number of planted Cayley trees at each vertex of the core.  Then we study the 
case when the core is a regular graph (all vertices of the same degree) and we graft 
planted Cayley trees on the core vertices so that the resulting asymptotic Cayley tree is also regular.

\subsection{The core is a complete graph}

Let $C$ be the complete graph on $n$ vertices which has spectrum $\{ 1 , -1/(n-1) \}$.
The positive eigenvalue is simple, and the eigenfunction is constant on $C$. The negative 
eigenvalue has multiplicity $n-1$.
The  eigenfunctions can be chosen to take the value $-(n-1)$
on one vertex, which can be any vertex on $C$, and $1$ on all other vertices. Of the $n$ 
such functions, any $n-1$ are linearly independent.
Now graft
$p$ copies of  $T'_q$  to each vertex of $C$.  The resulting graph $G$ is characterized 
by the 3 integers $p,q,n$.  
The permutation symmetry of the vertices of $C$ ensures that the eigenfunctions of $K^G$ take the same form on $C$ as those of $K^C$.  

In the planted Cayley trees the eigenfunction $\psi$ of  $K^G$  with eigenvalue $\lambda$ is given by \eqref{le4} and \eqref{le5}.
For the eigenfunction that is constant on $C$, the eigenvalue equation for a vertex on $C$ is 
\beq \label{cg2}
p\xi^{-1}+n-1=(p+n-1)\lambda\,.
\eeq
Combining \rf{le5} and \rf{cg2} shows that $\xi=1$ or
\beq\label{cg3}
\xi^{-1}=\frac{p+n-1}{ q(n-1) -(p+n-1)}\, .
\eeq
We see that $\psi$ is in $\ell^2(G)$ if 
\beq
1+\sqrt{q-1} <\frac{q(n-1)}{ p+n-1}\,, \label{cg4}
\eeq
which is always the case for large enough $q$ if $n$ and $p$ are fixed. From \eqref{le5} the corresponding value of 
$\lambda$ then lies  outside the 
continuous spectrum.

In the case of the degenerate eigenfunctions,  \eqref{cg2} is replaced by
\beq p\xi^{-1}-1=(p+n-1)\lambda\,,
\eeq
and we obtain
\beq
\xi^{-1}=\xi_\pm^{-1}=\frac{-q\pm \sqrt{q^2-4q(n-1)(p+n-1)+4(p+n-1)^2}}{ 2(q(n-1)-(p+n-1))}\,.\label{eqn:xicomplete}
\eeq
If $n$ and $p$ are fixed, and $q$ is large enough, then  $(q-1)\xi_+^{-2}<1$ and $\psi$ is in $\ell^2(G)$.
In particular,   if
\begin{align}
    q> 4(n-1+\frac{p^2-1}{n-1})(n+p-1)
\end{align}
then, uniformly in $n\ge 2$ and $p\ge 1$, the non-constant eigenfunctions of $K^C$ correspond to $\ell^2$-eigenfunctions of $K^G$.
This is in contrast to the eigenstate that is constant on the core which persists for all $n\ge 4$ and all $q\ge 3$.

\subsection{$q$-Regular asymptotic trees}
Theorem \ref{thm:continuous} applies uniformly to all asymptotic Cayley trees but Theorems 
\ref{thm:decoupled} and \ref{thm:point} are strongest for those whose number of graft vertices is 
small relative to the total number of vertices in $G^{(0)}$. By contrast we can make more 
detailed statements about $S^G_{pp}$ for asymptotic Cayley trees consisting of equal numbers 
of planted Cayley trees grafted to every vertex of a core that is a regular graph.
\begin{thm}\label{thm:qregular}
Let $G^{(0)}$ be a $(q-p)$-regular graph, where $q>p+2$, and let 
$G$ be the $q$-regular asymptotic tree 
obtained by grafting $p\ge 1$ copies of $T'_q$ at each vertex $v\in G^{(0)}$. Then the following hold:
\begin{itemize}
    \item[(i)] if $p>q-1-\sqrt{q-1}$ then 
    $S^{G}_{pp}$ is empty;
    \item[(ii)] if $p\le q-1-\sqrt{q-1}$ there is at least one member of $S^{G}_{pp}$.
\end{itemize}
\end{thm}

\begin{proof}
For a walk in $G$ the step from a vertex $u\in G^{(0)}$ to a neighbouring vertex $v\in G^{(0)}$ 
can be decomposed into excursions from $u$ into the attached trees followed by a final step along 
the edge $(u,v)$. As all vertices in $G$ have the same degree, it follows that, for  $u\in G^{(0)}$,
\begin{align}
    Q^G_u(z)=Q^{G^{(0)}}_u(f(z))\,, \label{eqn:polereln}
\end{align}
where
\begin{align}
    f(z)=\frac{z}{1+\frac{p\, (1-P(z))}{q-p}}\,.\label{eqn:freln}
\end{align}
By the spectral theorem, the poles of $Q^{G^{(0)}}_v(z)$ are at $\lambda^{-1}$, 
where $\lambda\in S^{G^{(0)}}\setminus \{0\}$ and those of $ Q^G_v(z)$ are at $\lambda'^{-1}$, 
where $\lambda'\in S^{G}_{pp}\setminus \{0\}$. It follows from \eqref{eqn:polereln} that if 
$\lambda\in S^{G^{(0)}}\setminus \{0\}$, and 
a real solution $\lambda'^{-1}\in[-1,1]$ exists to $\lambda^{-1}=f(\lambda'^{-1})$, then 
$\lambda'\in S^{G}_{pp}$. From \eqref{eqn:freln} we obtain
\begin{align}
    \lambda=\frac{\lambda'q(2q-2-p)+\mathrm{sgn}(\lambda')\,p\sqrt{q^2\lambda'^2-a^2}}{2(q-1)(q-p)}\,,
    \label{eqn:lambdareln}
\end{align}
and see immediately that there are real solution pairs $\lambda,\lambda'$  only if $\abs{\lambda'}>
\frac{a}{q}$ and $\abs{\lambda}>\frac{2q-2-p}{(q-p)\sqrt{q-1}}$. If
$p>q-1-\sqrt{q-1}$ the second condition becomes $\abs{\lambda}>1$ which is not possible so statement 
$(i)$ is proved. On the other hand, if $p\le q-1-\sqrt{q-1}$, straightforward algebra shows that 
there are solution pairs to \eqref{eqn:lambdareln} in which $s\in [0,\log\frac{q-p-1}{\sqrt{q-1}} ]$ 
is computed from
\begin{align}\abs\lambda &= 2\,\frac{\sqrt{q-p-1}}{q-p}\cosh\left(s+\tanh^{-1}\frac{p}{2q-p-2}\right)\,,
\label{eqn:lambdaG0}
\end{align}
and then $\lambda'$ is given by
\begin{align}
    \lambda'&=\mathrm{sgn}(\lambda)\, \frac{a}{q}\cosh (s)\,. \label{eqn:lambdaG}
\end{align}
In particular  $S^{G^{(0)}}$ always contains $\lambda=1$,  and we obtain in that case  
\begin{align}
    \lambda'=1-\frac{p(q-p-2)}{q(q-p-1)}\,,\label{lambdaprime}
\end{align}
which proves statement $(ii)$. Note further that for every eigenvalue $\lambda\in S^{G^{(0)}} $ 
satisfying these conditions there is at least one  eigenstate of $K^G$  with corresponding 
eigenvalue $\lambda'$ given by using \eqref{eqn:lambdaG0} and \eqref{eqn:lambdaG}. 
\end{proof}

We note that if $G^{(0)}$ is a complete graph on $n$ vertices then $G$ is $q$-regular if
$n+p-1=q$.  In that case Theorem \ref{thm:qregular}($ii$) gives the inequality \rf{cg4}.

Theorem \ref{thm:qregular} is easily generalised to the graphs for which the degree of vertices in  
$G^{0)}$ differs from $q-p$ but we will not give details here. On a regular graph $K^G$ is simply 
proportional to the adjacency matrix. 

Our result is therefore related to those of \cite{McKay} which 
considers the class $\mathcal C_q$ of $q$-regular graphs 
such that 
the density of closed loops inside a ball of radius $R$ goes to zero as $R$ goes to infinity.
Clearly, $q$-regular asymptotic Cayley trees are a special case of this 
construction. In \cite{McKay}  it is shown that the continuous spectrum of a graph $G\in \mathcal C_q$ 
is the same as that of $T_q$ and the limiting spectral density is \eqref{eqn:spectral}; this result can 
be viewed as a coarse-grained version of Theorem \ref{thm:continuous}.

\section{Quantum walks}\label{sec:quantumwalks}

\subsection{Definitions}

As an example of the application of our results we consider a class of continuous time  quantum walks. 
Let $H$ be a self-adjoint operator 
on $\ell^2(G)$ with respect to the inner product $(\cdot, \cdot)$.
A quantum walk on $G$ is then defined by the unitary time development operator
\beq\label{2}
U(t)=e^{-itH}\,.
\eeq
If the state of the walk at time $t=0$ is given by the normalized vector $|\psi_0\kt$, then its state at
time $t$ is $|\psi (t)\kt =U(t)|\psi_0\kt$.  In particular, if the walk is at the vertex $u$ at time $0$, then the 
probability amplitude that it is at a vertex $v$ at time $t$ is given by
\beq\label{3}
A_t(u,v) = \sqrt{\frac{\sigma_u}{\sigma_v}}\br v|U(t) |u\kt\,,
\eeq
and the probability of finding the walk at $v$ at time $t$, given that it was at $u$ at time $t=0$ is
\beq\label{4}
P_t(u,v)=|A_t(u,v)|^2\,.
\eeq

For classical random walks we define the spectral dimension of a graph $G$
to be $d_s$ if, and only if, $p_n(u,u)\sim n^{-d_s/2}$ as $n\to\infty$. It can be shown that $d_s$, if it exists, does not depend on $u$.   If $p_n(u,u)$ falls off faster than any power of $n$ then one says that the spectral dimension is infinite.  
From \eqref{rw2} it follows that
\begin{itemize}
\item[(i)] If $Q_u(z)$ is analytic at $z=1$ then $p_n(u,u)$ decays at least exponentially with $n$ so $d_s=\infty$. Using \eqref{app3} we see that this is the case for a regular Cayley tree.
\item[(ii)] If $Q_u(z)$ has a branch point of the form $(1-z)^{\beta}$, then, by a Tauberian theorem \cite{flajolet:2009}, $d_s = 2\beta+2$. This is the case on, for example, the graph $\bbZ$ (equivalently the regular Cayley tree with $q=2$) which has $d_s= 1$, as can be seen by setting $q=2$ in 
\eqref{app3}. In \cite{DJW1} the
 generating function and its corresponding spectral dimension are calculated for a number of simple graphs.
\end{itemize}
Analogously 
the quantum 
spectral dimension $d_{qs}$ of a graph $G$ is defined by
 \beq\label{rw5}
 P_t(u,u)\sim t^{-d_{qs}/2},~~t\to\infty \,.
 \eeq
 For example, in the case of $\bbZ$, and with $H=-K^G$,
  it is easy to see that $d_{qs}=2$. It can be shown quite generally that  $d_{qs}\le 2d_s$ \cite{davidjonsson}. We will see below that this upper bound on $d_{qs}$ is not necessarily saturated; indeed on asymptotic Cayley trees
 $d_{qs}=6$.

\subsection{Quantum walk on an asymptotic Cayley tree} 

Our results on the spectrum of asymptotic Cayley trees can be used to understand the properties of the quantum walks described above on such graphs.

\begin{thm}\label{thm:asympt}
    Let $G$ be an asymptotic Cayley tree, and $H=-K^G$. 
    Quantum walk on $G$ has the following properties for large elapsed time $t$.
    \begin{itemize}
\item[(i)] If $S_{pp}^G=\emptyset$, then the probability of staying inside a fixed ball around the starting vertex $u$  decays at a rate $t^{-3}$ for large $t$  and hence $d_{qs}=6$. 
(Note that this includes the case of the   pure Cayley tree $T_q$.)
\item[(ii)] If $S_{pp}^G\subset \cut$, then 
a walker starting at   $u\notin G^{(0)}\setminus V^G_0 $ stays within a fixed ball around $u$ with a probability decaying asymptotically as $t^{-3}$ for $t$ large. On the other hand, a walker starting inside $G^{(0)}\setminus V_0^G$ is trapped in this set with a probability $P_t(u)$ at time $t$ of the form 
\beq
P_t(u) = \sum_{\lambda\in S^G_{pp}} \br u|e_\lambda|u\kt + \tilde C^G_t(u)\,t^{-\frac 32} + O(t^{-\frac 52})\,,
\eeq
where $\tilde C^G_t$ is a bounded oscillating function of $t$ (and we note that the $t$-independent term on the right-hand side is positive unless all eigenfunctions of $K^G$ vanish at $u$).
\item[(iii)] If $S_{pp}^G$ contains at least one eigenvalue $\lambda\notin\cut\cup S^{G^{(0)}}$, then 
there is at least one tree $T_q'$ grafted to $G^{(0)}$ at some vertex $v_0$ such that 
a walker starting at $u\in T_q'$ will be trapped in a ball of fixed radius $R>0$ around $u$ 
with a probability that declines exponentially with $d^G(v_0,u)$ in the limit $t\to\infty$.
\end{itemize}
\end{thm}

\noindent To prove these results we need the following lemma.
\begin{lem}\label{lemma:Atdep}
 Let $G$ be an asymptotic Cayley tree, and let $H=-K^G$. Then the behaviour at large $t$ of the probability amplitude $A^G_t(u,v)$ for finite $d_G(u,v)$ is given by
    \begin{align}
    A^G_t(u,v)&= \sum_{\lambda\in S_{pp}^G }    
    \sqrt{\frac{\sigma_u}{\sigma_v}}\br v| e_\lambda| u\kt\, e^{i\lambda t}+C^G_t(u,v) \,t^{-{3/2}}+ O(t^{-5/2})\,, \label{eqn:AGt}
\end{align}
where $C^G_t$ is a bounded oscillating function of $t$ and $e_\lambda$ is the spectral projection of $K^G$ corresponding to $\lambda$.
\end{lem}
\begin{proof}
    By \eqref{eqn:specrepgen} the probability amplitude $A^G_t(u,v)$ is given by
\begin{equation}
A^G_t(u,v)=\sqrt{\frac{\sigma_u}{\sigma_v}}\sum_{\lambda\in S^G_{pp}}\br v|e_\lambda|u\kt\, e^{it\lambda}
- \sqrt{\frac{\sigma_u}{\sigma_v}}\frac{1}{\pi}\int_{-\frac aq}^{\frac aq} e^{it\lambda} {\tilde B}^G_{u,v}(\lambda)\sqrt{a^2 -q^2\lambda^2}\,d\lambda \,,\label{eqn:Ain ACT}
\end{equation}
where  $\tilde B^G_{u,v}(\lambda)$ is a rational function with no poles in $\cut$. 
Setting $\lambda=\frac aq \cos\theta$ the integral term in \eqref{eqn:Ain ACT} becomes 
\beq
A^G_{u,v}(t)_{cont} = -\sqrt{\frac{\sigma_u}{\sigma_v}}\frac{a^2}{q\pi} \int_{0}^{\pi} e^{i\tfrac aq t\cos\theta}\tilde B^G_{u,v}(\tfrac aq \cos\theta) \sin^2\theta\, d\theta\,.
\eeq
By partial integration this gives
\begin{align}
A^G_{u,v}(t)_{cont} =&\nonumber\\ - \sqrt{\frac{\sigma_u}{\sigma_v}}&\frac{a}{i\pi t} \int_{0}^{\pi} e^{i\tfrac aq t\cos\theta}\left(\tilde B^G_{u,v}(\tfrac aq \cos\theta) \cos\theta - \frac aq\frac{d{\tilde B}^G_{u,v}}{d\lambda}(\tfrac aq \cos\theta)\sin^2\theta\right)\, d\theta\,.
\end{align}
Performing a further integration by parts of the contribution from the second term in 
parentheses and applying a standard stationary phase approximation we get 
\begin{align}
    A^G_t(u,v)_{cont}=-\sqrt{\frac{aq\sigma_u}{2\pi t^3\sigma_v}} \left(e^{i(\frac{at}{q}+\frac \pi 4)} \tilde B^G_{u,v}({\textstyle\frac aq})- e^{-i(\frac{at}{q}+\frac \pi 4)} \tilde B^G_{u,v}(-{\textstyle \frac aq})\right)+ O(t^{-5/2})\,.\label{eqn:aCTcont}
\end{align}
Combining \eqref{eqn:Ain ACT} and \eqref{eqn:aCTcont} proves the lemma.
\end{proof}

\begin{proof}[Proof of Theorem \ref{thm:asympt}]
The probability that a walker is to be found within a distance $R$ of the starting vertex 
$u\in G$ after time $t$ has elapsed
is
\begin{align}
    P_t(R,u)=\sum_{v:\, d_G(v,u)\le R} \abs{A_t(u,v)}^2\,.
\end{align}
First, let  $S_{pp}^G=\emptyset$. Then, 
 by Lemma \ref{lemma:Atdep}, 
  \begin{align}
    P_t(R,u)=\frac{1}{t^3}\sum_{v:\, d_G(v,u)\le R} \abs{C^G_t(u,v)}^2 +O(t^{-4})\,,\label{eqn:Tfree}
\end{align}
but $C^G_t(u,v)$ is a bounded oscillating function of $t$ 
  so statement $(i)$ follows.

Now let $S_{pp}^G\subset\cut$. 
If $\lambda\in S_{pp}^G$ and  $u\notin V(G^{(0)})\setminus V^G_0$, then by 
Lemma \ref{lemma:sphericalLe}$(i)$ $\br v|e_\lambda|u\kt$ vanishes for all $v\in V(G)$ 
so \eqref{eqn:Tfree} holds by Lemma \ref{lemma:Atdep}. On the other hand, if $u\in V(G^{(0)})\setminus V^G_0$ then $\br v|e_\lambda|u\kt=0$ for $v\notin V(G^{(0)})\setminus V^G_0$ and the probability that the walker stays inside $G^{(0)}$ at 
time $t$ equals 
\begin{align}
    P_t(u)&=\sum_{v\in V(G)\setminus V^G_0} \left|\sum_{\lambda \in S_{pp}^G} \sqrt{\frac{\sigma_u}{\sigma_v}}\br v|e_\lambda| u\kt  e^{i\lambda t}\right|^2\nonumber \\
    &+2t^{-3/2}\mathrm{Re}\sum_{\mu \in S_{pp}^G}\sum_{v\in G^{(0)}\setminus V^G_0} \sqrt{\frac{\sigma_u}{\sigma_v}}\br v|e_\lambda| u\kt C^G_t(u,v)^*e^{i\lambda t} +O(t^{-5/2})\,.
\end{align}
The first term evaluates to the constant $\sum_{\lambda\in S^G_{pp}} \br u| e_\lambda| u\kt $, 
and statement $(ii)$ follows.

Finally, if $S_{pp}^G$ contains at least one eigenvalue $\lambda\notin\cut\cup S^{G^{(0)}}$ then it follows from (the proof of)  Lemma \ref{lemma:sphericalLe}$(ii)$ that there is at least one vertex $v_0\in V^G_0$ such that $\br v_0|e_\lambda|v_0\kt >0$ and if $u,v$ belong to one of the $p$ trees $T_q'$ grafted at $v_0$ we have 
\begin{equation}
    \br v|e_\lambda |u\kt = \frac{\sigma_{v_0}+p}{q}\br v_0|e_\lambda|v_0) \,(\lambda P(\lambda^{-1}))^{d^G(v_0,u)+d^G(v_0,v)}
\end{equation}
(where, as previously, $\sigma_{v_0}$ denotes the degree of $v_0$ in $G^{(0)}$).
Since this identity holds for all eigenvalues $\lambda\notin\cut$ and $\br v|e_\lambda| u\kt = 0$ if $\lambda\in\cut$, we obtain from \eqref{eqn:AGt}, for any vertex $u\in T'_q$ with $d^G(v_0,u)>R$, that
\begin{align}
 P_t(R,u)=\sum_{v:d^G(u,v)\leq R}\Big|\sum_{\lambda\in S_{pp}^G\setminus\cut} \tfrac{\sigma_{v_0}+p}{q}\br v_0|e_\lambda|v_0)\,e^{i\lambda t} &\Big(\lambda P(\lambda^{-1})\Big)^{d^G(v_0,u)+d^G(v_0,v)}\Big|^2 \nonumber\\ &+ O(t^{-3/2})\,.
\end{align}
Here, the first term on the right-hand side is a non-vanishing bounded oscillating function of $t$ and, recalling that $|\lambda P(\lambda^{-1})|< (q-1)^{-\frac 12}$, if $|\lambda|>\frac aq$, it decays for fixed $R$ exponentially with $d^G(v_0,u)$. This proves statement $(iii)$.
\end{proof}

In the case of $T_2$ (i.e. the infinite chain) $S^G_{pp}=\emptyset$;  using \eqref{app6} and setting  $q=2$, it is easily seen that \eqref{eqn:Ain ACT}
reduces to the well known formula \cite{fahri}
\beq\label{a6A}
  A^{T_2}_t(u,v)= \frac{1}{2\pi}\int_{-\pi}^\pi e^{it\cos\theta +i\ell\theta}\,d\theta =
i^{\ell}J_{\ell}(t),
\eeq
where $J_{\ell}$ is the $\ell$-th Bessel function.

A similar analysis can be used to show that a walker who escapes into a grafted tree always behaves ballistically at large times. Here we outline the calculation to demonstrate this.  Let $G$ be an asymptotic Cayley tree obtained by grafting $p\geq 1$ trees $T'_q$ at a vertex $v_0$ of a graph $B$ (which hence is either finite or an asymptotic Cayley tree), and let $u\in B$. 
Let $\nu>0$  be a real number and consider times $t$  chosen so that $\nu t$ is an integer.
The probability that the walk is at any  vertex $v_{\nu t}$ in the tree(s) a distance $ \nu t$ from $v_0$ is then given at large times by
\begin{align}
    P_t(v_{\nu t},u)=(q-1)^{ \nu t -1}\abs{A^G_t(u,v_{\nu t})}^2\,.
\end{align}

We have 
from \eqref{eqn:Ain ACT},\eqref{eqn:QG2}  and \eqref{eqn:expdecay} that
\begin{align}
A^G_t(u,v_{ \nu t})=
  \sqrt{\frac{\sigma_u}{q}} &\sum_{\lambda \in S^G_{pp}\setminus\cut}    \br v_0\vert e_\lambda|u\kt\,  ( \lambda P(\lambda^{-1}))^{ \nu t}\hfill\nonumber\\
& -\sqrt{\frac{\sigma_u}{q}}\frac{a^2}{q\pi} \int_{0}^{\pi} e^{i(\tfrac aq t\cos\theta+ \nu t\theta)} B_{u,v_0}(\tfrac aq \cos\theta) \sin^2\theta\, d\theta\,,
\end{align}
where
\begin{align}
    B_{u,v_0} (\lambda)=\lambda^{-1} \frac{q\,Q(\lambda^{-1})Q^B_{u,v_0}(\lambda^{-1})}{\sigma_{v_0}Q(\lambda^{-1}) + p\,Q^B_{v_0}(\lambda^{-1})} \,. 
\end{align}
It follows from Theorem \ref{thm:continuous} that, since $\lambda P(\lambda^{-1})$ has no poles and no zeroes,  $B^G_{u,v_0} (\lambda)$ is analytic for $\lambda \notin \cut\cup S^G_{pp}$. 
For large $t$, the integral  can be evaluated  by 
the method of steepest descent. If  $\nu=\frac{a}{q}\sin\theta_0$, $\theta_0\in (0,\frac{\pi}{2})$ we obtain
 \begin{equation}
         A^G_t(u,v_{ \nu t})_{cont}=i\nu\sqrt{\frac{\sigma_u}{2\pi qt\lambda_+}} \,\left(\frac{-1}{q-1}\right)^{\frac{1}{2}\nu t}\, 
   \left(B_{u,v_0}(\lambda_+)\, e^{i\phi_t} +B_{u,v_0}(-\lambda_+)\, e^{-i\phi_t}\right)\,,
         \end{equation}
         where $\lambda_+=\frac{a}{q}\cos\theta_0$ and
         \begin{equation}
              \phi_t=t\left(\lambda_++\nu(\theta-\frac{\pi}{2})\right)-\frac{\pi}{4}\,.
         \end{equation}
 The contribution from $S^G_{pp}$ decays exponentially with $t$, so  $P_t(v_{\nu t},u)$
 is an oscillating function whose magnitude 
 decays like $t^{-1}$. For $\nu=\frac{a}{q}\cosh\theta_0$, $\theta_0\in (0,\infty)$, we obtain  
    \begin{equation}
         A^G_t(u,v_{ \nu t})_{cont}= i\nu\sqrt{\frac{\sigma_u}{2\pi qt\abs{\lambda}}} \left(\frac{-1}{q-1}\right)^{\frac{1}{2}\nu t} 
          \, 
   B_{u,v_0}(\lambda)\, e^{-\phi_t} \,,
         \end{equation}
         where $\lambda=-i\frac{a}{q}\sinh\theta_0$ and 
         \begin{equation}
              \phi_t=t\left(\nu\theta_0 -\abs{\lambda}\right)\,,
         \end{equation}
so all contributions to  $P_t(v_{\nu t},u)$ decay exponentially with $t$.  Hence the walker's motion is  ballistic at large times, and there is an Airy type front moving with  velocity $\nu_c=\frac{a}{q}$.  This behaviour is qualitatively very similar to the $q=2$ case \cite{fahri}, and to that of continuous time quantum walk on the infinite toothed comb \cite{davidjonsson}.

In summary, the motion on a pure Cayley tree is ballistic and, as $t\to\infty$,  the walker moves towards infinity at a constant 
rate. The  probability for return to the 
starting point after time $t$ decays like $t^{-3}$. This is in contrast with the  
exponential decay of the classical case; on the highly branched structure  of a Cayley tree, quantum walkers remain 
in the vicinity of their starting point \emph{longer} than 
do classical walkers.
  In the case of an asymptotic Cayley tree
a new feature occurs if $S_{pp}^G\ne \emptyset$:  the square summable eigenfunctions prevent the 
walk from moving to infinity with probability one so the walker is at least partially trapped in the vicinity of the core.  Whether such eigenfunctions exist depends on the structure of the core graph but we have shown  that in general they do, and given some explicit examples.
The part of the wave function which is not trapped behaves
qualitatively as on the pure Cayley tree; so a walker who escapes from the vicinity of the core moves towards infinity at a constant rate.
The core graph therefore behaves like a trap and the resulting properties of quantum walk are qualitatively similar to those observed 
on graphs with a trap potential on a single site \cite{Luck}.

\bigskip
\noindent{\bf Acknowledgements} JFW's research was funded by Research England. BD acknowledges support from Villum Fonden 
via the QMATH Centre of Excellence (Grant no. 10059).  TJ would like to acknowledge hospitality at the Rudolf Peierls
Centre for Theoretical Physics in Oxford and at the Department of Mathematical Sciences in the University of Copenhagen.

\end{document}